\documentclass[a4paper,11pt]{amsart}
\usepackage{amsmath,amssymb}

\usepackage{amsfonts}
\usepackage{eucal}
\usepackage{amsthm}
\usepackage{graphicx}
\usepackage{framed}
\usepackage{color}

\newcommand{\bigpare}[1]{\bigl(#1\bigr)}
\newcommand{\biggpare}[1]{\biggl(#1\biggr)}

\newcommand{\biggbra}[1]{\biggl\{#1\biggr\}}

\newcommand{\biggbrac}[1]{\biggl[#1\biggr]}

\newcommand{\bigset}[2]{\bigl\{#1\bigm|#2\bigr\}}

\newcommand{\Biggset}[2]{\Biggl\{#1\Biggm|#2\Biggr\}}
\newcommand{\norm}[1]{\| #1 \|}
\newcommand{\bignorm}[1]{\bigl\| #1 \bigr\|}

\newcommand{\abs}[1]{| #1 |}
\newcommand{\bigabs}[1]{\bigl| #1 \bigr|}

\newcommand{\Biggabs}[1]{\Biggl| #1 \Biggr|}
\newcommand{\jap}[1]{\langle #1 \rangle}

\def\a{\alpha}

\def\f{\varphi}
\def\g{\psi}

\def\l{\lambda}

\def\n{\nu}

\def\s{\sigma}

\newcommand{\F}{\Phi}
\newcommand{\G}{\Psi}
\newcommand{\dd}{\mathrm{d}}

\renewcommand{\L}{\mathcal{L}}
\renewcommand{\H}{\mathcal{H}}

\newcommand{\V}{\mathcal{V}}

\def\re{\mathbb{R}}
\def\co{\mathbb{C}}
\def\ze{\mathbb{Z}}

\newtheorem{thm}{Theorem}[section]
\newtheorem{lem}[thm]{Lemma}

\newtheorem{cor}[thm]{Corollary}

\theoremstyle{definition}

\newtheorem{ass}{Assumption}

\theoremstyle{remark}

\numberwithin{equation}{section}

\title{Continuum limit of the lattice quantum graph Hamiltonian}
\author{Pavel Exner}
\address{P. Exner: Doppler Institute for Mathematical Physics and Applied Mathematics, Czech Technical University, B\v rehov\'a 7, 11519 Prague, and Department of Theoretical Physics, Nuclear Physics Institute, Czech Academy of Sciences, 25068 \v{R}e\v{z}, Czechia}
\email{exner@ujf.cas.cz}
\author{Shu Nakamura}
\address{S. Nakamura: Department of Mathematics, Faculty of Sciences, Gakushuin University, 1-5-1, Mejiro, Toshima, Tokyo, Japan 171-8588}
\email{shu.nakamura@gakushuin.ac.jp}
\author{Yukihide Tadano}
\address{Y. Tadano: Department of Mathematics, Faculty of Science Division I, Tokyo University of Science, 1-3,
Kagurazaka, Shinjuku, Tokyo, Japan 162-8601}
\email{y.tadano@rs.tus.ac.jp}

\begin{document}
\maketitle

\begin{abstract}
We consider the quantum graph Hamiltonian on the square lattice in Euclidean space, and we show that the spectrum of the Hamiltonian converges to the corresponding Schr\"odinger operator on the Euclidean space in the continuum limit, and that the corresponding eigenfunctions and eigenprojections also converge in some sense. We employ the discrete Schr\"odinger operator as the intermediate operator, and we use a recent result by the second and third author on the continuum limit of the discrete Schr\"odinger operator.
\end{abstract}

\section{Introduction}

In mathematics and physics one often meets situations when we investigate a large structure being interested in its gross properties independent of the local structure. A classical example is the homogenization theory -- see, e.g., \cite{AP, BP, BS} and references therein. Another example providing deep mathematical problems as well as a number of applications is represented by properties of large networks \cite{LL}.

The present paper is devoted to a problem od this type appearing in the theory of \emph{quantum graphs}, which is a short name for Schr\"odinger operators the configuration space of which is a metric graph \cite{BK}. To make such an operator self-adjoint, it is not enough to have the potential real-valued and sufficiently regular; one also has to define properly the conditions matching functions from the operator domain at the graph vertices \cite[Thm.~1.4.4]{BK}. There is a large number of ways how to choose those conditions among which there is a smaller and distinguished subset, namely those preserving continuity at the vertices. In such a case there is just one real parameter associated with each vertex; usually the term \emph{$\delta$-coupling} is employed.

This paper is concerned with a family of such quantum graphs. It was observed in~\cite{ExHeSe} that a square lattice graph with a varying $\delta$-coupling at the vertices and the vertex spacing tending to zero can approximate Schr\"odinger operator in $L^2(\re^\nu)$ provided the energy is rescaled by the dimension $\nu$; this approximation was illustrated on chaotic motion in billiards. What was left out there, however, was the existence of the limit and the type of the convergence. These are the question addressed here. To get the answer we combine two main elements. One is the recent result of two of the present authors \cite{NaTad} on the continuum limit of discrete Schr\"odinger operators, the other is the duality \cite{Ca, Ex, Pa} between a Schr\"odinger operator on a metric graph and a suitable operator on the associated discrete graph.

\section{Problem statement and the main result}

To begin with, we introduce in the standard way \cite{BK} the quantum graph Hamiltonian, that is, the Schr\"odinger operator on the metric graph. The latter will be in our case the $\n$-dimensional square lattice graph of the lattice spacing $\ell>0$,
\[
\Gamma= (\V,\L), \quad \V=\ell\ze^\n, \quad
\mathcal{L}=\bigset{\L_{jn}=[j,n]}{j,n\in \V, |j-n|=\ell}
\]
where $[j,n]$ denotes the line segment connecting $j$ and $n\in\ell\ze^\n$. The symbol $\V$ denotes the set of vertices, and $\L$ is the set of edges in $\Gamma$. We introduce the Hilbert space of functions on the graph by
\[
\H_1= L^2(\Gamma)=\bigoplus_{\L_{jn}\in\L} L^2(\L_{jn}),
\]
with the inner product
\[
\jap{\f,\g}_{\H_1} = \frac{\ell^{\n-1}}{\n} \sum_{\L_{jn}\in\L}
\int_{\L_{jn}} \overline{\f_{jn}(t)}\g_{jn}(t)\,\dd t, \quad\text{where }\:
\f=(\f_{jn}), \g=(\g_{jn})\in \H_1.
\]
We adopt the following hypothesis:
\begin{ass}\label{ass-main}
$V$ is a real-valued continuous function on $\re^d$, and bounded from below. Furthermore, $(V(\cdot)+M)^{-1}$ is uniformly continuous for some $M>0$, and there is a $c_1>0$ such that
\[
c_1^{-1} (V(x)+M) \leq V(y)+M\leq c_1(V(x)+M) \quad \text{if }\: |x-y|\leq 1.
\]
\end{ass}

\medskip

We denote $V_j=V(j)$ for $j\in\V$ and set $\a_j:=\ell V_j$ for $j\in\V$. The Sobolev space of order one on the graph $\Gamma$ is then given by
\[
H^1(\Gamma)=\bigset{(\f_{jn})\in \H_1}{\f_{jn}\in H^1([j,n]),
 \f_{jn}(j)=\f_{jm}(j)\text{ for }\: j\in\V \text{ and }n,m\in\V(j)},
\]
where $\V(j)=\{n\,|\, |j-n|=\ell\}$ is the set of vertices adjacent to $V_j$, in other words, the neighborhood of the point $j$ in the discrete graph associated with $\Gamma$.

On the domain $\mathcal{Q}(H_1)=\bigset{\f\in H^1(\Gamma)}{\sum_{j\in\V}\a_j|\f_j|^2<\infty}$ we define the quadratic form $q_\a$ by means of the formula
\[
q_\a(\f,\g)=\jap{\f',\g'} + \sum_{j\in\V}\a_j\overline{\f_j}\g_j,\quad \f,\g\in \mathcal{Q}(H_1),
\]
where $(\f')_{jn}(t)=\frac{\dd}{\dd t}\f_{jn}(t)$ on $\L_{jn}$, and $\f_j=\f_{jn}(j)$ for $n\in\V(j)$. We denote the self-adjoint operator associated with $q_\a$ by $H_1$, that is, $\jap{\f,H_1\g} = q_\a(\f,\g)$
holds for $\f,\g\in \mathcal{D}(H_1)$. It is known \cite[Sec.~1.4.3]{BK} that
\[
\mathcal{D}(H_1)=\Biggset{\g=(\g_{jn})\in H^1(\Gamma) \cap \bigoplus_{\L_{jn}\in\L} H^2(\L_{jn})}{\sum_{n\in\V(j)}\g_{jn}'(j)=\a_j\g_j}
\]
and $(H_1\g)_{jn}(t)=-\g_{jn}''(t)$ on $\L_{jn}$. We recall that $\Gamma$ is regarded as a non-oriented graph and the derivatives entering the condition specifying $\mathcal{D}(H_1)$ are all conventionally taken in the outward direction.

The second object to consider is the Schr\"odinger operator $H$ on $L^2(\re^\n)$ given by
\[
H\f(x)=-\triangle\f(x) +V(x)\f(x), \quad x\in\re^\n
\]
for $\f\in \mathcal{D}(H)=\bigset{\f\in H^2(\re^\n)}{V\f\in L^2(\re^\n)}$. We recall that under our assumption about the potential, $H$ is a self-adjoint operator on $L^2(\re^\n)$.

Our main result claims that in the limit $\ell\to 0$ the operators $\n H_1$ approximate $H$ in the sense of norm resolvent convergence.
\begin{thm}\label{thm-main}
Let $z\in \co\setminus\re$, and adopt Assumption~\ref{ass-main}. Then there is a bounded operator $\G:\:\H_1\to L^2(\re^\n)$ such that in the limit $\ell\to 0$ we have
\begin{align*}
&\bignorm{(H-z)^{-1} - \G(\n H_1-z)^{-1} \G^*}_{\mathcal{B}(L^2)}\to 0,\\[.2em]
&\bignorm{(\n H_1-z)^{-1} - \G^*(H-z)^{-1}\G}_{\mathcal{B}(\H_1)}\to 0.
\end{align*}
\end{thm}

\medskip

Since the approximation relates operators acting in different Hilbert spaces, it is essential to specify a suitable identification map between them. Approximations of such a type have been studied in other situations \cite{EP, NaTad}, whereas a general presentation of the method can be found in the book \cite{Po}.

We will construct the operator $\G$ later, but we note already here that it is not invertible, and actually neither injective nor surjective. Still this result is sufficient to establish the convergence of the spectrum. We denote the spectrum of a self-adjoint operator $A$ by $\s(A)$, and by $E_A(\Omega)$ its spectral projection to $\Omega\subset\re$.
\begin{cor}\label{cor-spectral-projection}
Let $a,b\in\re$, $a<b$, and suppose $a,b\notin \s(H)$. Then $a,b\notin \s(\n H_1)$ if $\ell$ is
sufficiently small,  and
\[
\bignorm{\G E_{\n H_1}((a,b))\G^* -E_H((a,b))}_{\mathcal{B}(L^2)}\to 0 \quad \mathrm{as}\;\, \ell\to 0.
\]
In particular, if $\l$ is an isolated eigenvalue of $H$ of multiplicity $m$ and $\g\in L^2$ is a corresponding eigenfunction, $H\g=\l\g$, then there are isolated eigenvalues $\l_{k,\ell}$ and eigenfunctions $\g_{k,\ell}\in\H_1$ of $\n H_1$ for each $k=1,\dots,m$ and $\ell>0$, satisfying $\n H_1\g_{k,\ell}=\l_{k,\ell}\g_{k,\ell}$, such that $\l_{k,\ell}\to \l$ and $\sum_{k=1}^m\G\g_{k,\ell} \to \g$ as $\ell\to 0$.
\end{cor}

We also have the convergence of the spectrum of $\nu H_1$ to that of $H$ with respect to the Hausdorff distance, which is defined by
\[
d_\mathrm{H}(X,Y)=\max\biggbra{\sup_{x\in X} d(x,Y),\, \sup_{y\in Y} d(y,X)},
\]
where $X,Y\subset \re$ and $d(\cdot,\cdot)$ denotes the Euclidean distance.
\begin{cor}\label{cor-Hausdorff-convergence}
Let $M>0$ large enough to ensure that $-M<\inf\s(H)$. Then for all sufficiently small $\ell>0$ one has $-M\notin\s(\n H_1)$ and
\[
d_\mathrm{H}(\s((H+M)^{-1}),\s((\n H_1+M)^{-1}))\to 0 \quad \mathrm{as}\;\, \ell\to 0.
\]
In particular, $\s(\n H_1)$ converges to $\s(H)$ as $\ell\to 0$ locally in terms of the Hausdorff distance.
\end{cor}
These corollaries follow from Theorem~\ref{thm-main} in a simple way using arguments analogous to those employed in \cite{NaTad}.

\section{Discrete Schr\"odinger operator and its convergence}

To prove Theorem~\ref{thm-main} we choose an appropriate discrete Schr\"odinger operator as the intermediate object and use the recent result by two of the present authors \cite{NaTad} on its continuum limit, cf. also \cite{C-G-J} and \cite{IsJe} for fresh related results. Let us first recall the basic notions. The Hilbert space of functions on the vertices,
\[
\H_2= \ell^2(\V)=\ell^2(\ell\ze^\n),
\]
is equipped with the norm
\[
\norm{u}_{\H_2}^2 =\ell^\n \sum_{j\in \V} |u_j|^2, \quad \text{where }u=(u_j)\in \H_2.
\]
Let the potential $V=V(x)$ be as before. We denote again $V_j=V(j)$ for $j\in \mathcal{V}$ and define a discrete Schr\"odinger operator $H_2$ on $\H_2$ by
\[
(H_2\f)_j=-\triangle_d\f_j +V_j\f_j, \quad \triangle_d\f_j= \frac{1}{\ell^2}\sum_{n\in\V(j)}(\f_n-\f_j)
\]
for $\f=(\f_j)\in \H_2$; it is easy to check that $H_2$ is a self-adjoint operator with its domain $\mathcal{D}(H_2)=\bigset{u=(u_j)\in \H_2}{(V_j u_j)\in \H_2}$.

The following result was proved in \cite{NaTad}:
\begin{thm}[Nakamura-Tadano]\label{thm-NT}
Let $z\in \co\setminus\re$ and adopt Assumption~\ref{ass-main}. Then there is a bounded operator $\F$ :  $\H_2\to L^2(\re^\n)$ such that in the limit  $\ell\to 0$ we have
\begin{align*}
&\bignorm{(H-z)^{-1} - \F(H_2-z)^{-1} \F^*}_{\mathcal{B}(L^2)}\to 0, \\[.2em]
&\bignorm{(H_2-z)^{-1} - \F^*(H-z)^{-1}\F}_{\mathcal{B}(\H_2)}\to 0.
\end{align*}
\end{thm}

The identification operator $\F$ here is constructed using an orthonormal basis in $L^2(\re^\n)$, and it is an isometry from $\H_2$ into $L^2(\re^\n)$, see \cite{NaTad} for details. Thus we know that $H$ is approximated by the indicated discrete Schr\"odinger operator, and it will be sufficient to show that that the latter is in turn approximated by the quantum graph Hamiltonian, and vice versa.

\section{Approximation of the quantum graph Hamiltonian by the discrete Schr\"odinger operator}

Our aim is to show that $\n H_1$ and $H_2$ are close to each other with an appropriate identification map when the spacing $\ell$ is small. Throughout this section, we suppose that $V$ is bounded from below.

\subsection{Identification operators}

Let $I\, :\, \H_2\to \H_1$ be the embedding by linear interpolation, namely $\f=(\f_j)\in \H_1 \mapsto I\f=(\f_{jn})\in\H_2$ defined by
\[
\f_{jn}(x(t))= (1-t)\f_j+t\f_n, \quad \text{where } x(t)=(1-t)j+tn\in  [j,n].
\]
We note $I$ is bounded from $\H_2$ into $\H_1$.

Furthermore, we define the trace operator $K$ : $H^1(\Gamma)\to \H_2$ by
\[
K\,:\, \f=(\f_{jn})\in H^1(\Gamma) \mapsto (K\f)_j=\f_{jn}(j) \quad (\forall n\in \V(j)).
\]

%

\subsection{Preliminary estimates}

\begin{lem} \label{lem-IK-1-estimate}For any $\ell>0$ we have
\[
\norm{IK-1}_{\mathcal{B}(H^1(\Gamma),\H_1)}\leq \ell.
\]
\end{lem}
\begin{proof}
Given $\f=(\f_{jn})\in C^1(\Gamma)$, we write $\f_j=\f_{jn}(j)$, $\forall n\in\V(j)$. Let $\tilde\f=IK\f$, in other words
\[
\tilde\f_{jn}(x(t))= (1-t)\f_j+t\f_n,
\quad \text{where }\: x(t)=(1-t)j+tn, \ 0\leq t\leq 1.
\]
We identify $\L_{jn}\cong [0,\ell]$ for the moment. Then for $t\in [0,\ell]$ we have
\[
\f_{jn}(t)-\tilde\f_{jn}(t)=\int_0^t \f_{jn}'(s)\,\dd s -\frac{t}{\ell}\int_0^\ell \f_{jn}'(s)\,\dd s,
\]
since
\[
\tilde \f_{jn}(t)= \f_{jn}(j)+ \frac{t}{\ell}\int_0^\ell \f_{jn}'(s)\,\dd s, \quad t\in[0,\ell].
\]
From here we infer that
\[
\bigabs{\f_{jn}(t)-\tilde\f_{jn}(t)}\leq \int_0^\ell|\f_{jn}'(s)|\,\dd s
\leq \sqrt{\ell}\biggpare{\int_0^\ell |\f_{jn}'(s)|^2\,\dd s}^{1/2}
\]
using the Schwarz inequality, and consequently, we have
\[
\int_{\L_{jn}}\bigabs{\f_{jn}(t)-\tilde\f_{jn}(t)}^2\,\dd t
\leq \ell^2 {\int_{\L_{jn}} |\f_{jn}'(t)|^2\,\dd t}
\]
Summing up this over the edges $\L_{jn}$, we find
\[
\norm{\f-\tilde\f}_{\H_1}\leq \ell\norm{\f'}_{\H_1} \leq \ell\norm{\f}_{H^1(\Gamma)},
\]
and by the density argument, we get the estimate for any $\f\in H^1(\Gamma)$.
\end{proof}

\begin{lem}\label{lem-I^*-K-estimate} For any $\ell>0$ we have
\[
\norm{I^*-K}_{\mathcal{B}(H^1(\Gamma),\H_2)}\leq \frac{\ell}{\sqrt{5}}.
\]
\end{lem}
\begin{proof}
A simple computation yields for $\f\in C^1(\Gamma)$ the relation
\begin{align*}
(I^*\f)_j=& \frac{1}{\n\ell}\sum_{n\in \V(j)} \int_0^\ell \biggpare{1-\frac{t}{\ell}} \f_{jn}(t)\,\dd t \\[.2em]
=& \frac{1}{2\n}\sum_{n\in \V(j)} \biggbra{\biggbrac{-\biggpare{1-\frac{t}{\ell}}^2 \f_{jn}(t)}_0^\ell + \int_0^\ell \biggpare{1-\frac{t}{\ell}}^2 \f'_{jn}(t)\,\dd t} \\[.2em]
=& \f_j +\frac{1}{2\n}\sum_{n\in \V(j)} \int_0^\ell \biggpare{1-\frac{t}{\ell}}^2 \f'_{jn}(t)\,\dd t.
\end{align*}
This further implies
\begin{align*}
\norm{I^* \f - K\f}^2_{\H_2}=\,& \frac{\ell^\n}{4\n^2} \sum_{j\in\V}\,\Biggabs{\sum_{n\in \V(j)} \int_0^\ell \biggpare{1-\frac{t}{\ell}}^2 \f'_{jn}(t)\,\dd t\,}^2 \\[.2em]
\leq\,& \frac{\ell^\n}{2\n} \sum_{j\in\V}\sum_{n\in \V(j)} \int_0^\ell \biggpare{1-\frac{t}{\ell}}^4 dt  \int_0^\ell \abs{\f'_{jn}(t)}^2\,\dd t \\[.2em]
=\,& \frac{\ell^2}{5} \norm{\f'}_{\mathcal{H}_1}^2
\leq \frac{\ell^2}{5}\norm{\f}_{H^1(\Gamma)}^2,
\end{align*}
which also holds for any $\f\in H^1(\Gamma)$ proving thus the claim.
\end{proof}

Since $I$ is bounded, Lemmata~\ref{lem-IK-1-estimate} and \ref{lem-I^*-K-estimate} in combination with the triangle inequality give the following result:
\begin{cor}\label{cor-II^*-1-estimate} There is a $C>0$ such that for all $\ell>0$ we have
\[
\norm{II^*-1}_{\mathcal{B}(H^1(\Gamma),\H_1)}\leq C\ell.
\]
\end{cor}





\subsection{Explicit formula for $K(\n H_1-z)^{-1}I$}

Now we are going to derive an explicit expression for the sandwiched resolvent of the operator $\n H_1$ which will play the key role in the proof of our main theorem.
Following the standard convention, we write points of the resolvent set as $z=k^2$.
\begin{lem}\label{lem-key-ODE-estimate}
Let $\g=(\g_{jn})\in \big(\bigoplus_{jn} H^2(\L_{jn})\big)\cap H^1(\Gamma)$, $\f\in\H_2$, and  $k^2\notin \re$. Then there are operators $M_1,M_2\in \mathcal{B}(\H_2)$ such that
\begin{equation}\label{eq-H1-eigen-eq}
(\n H_1-k^2)\g=I\f
\end{equation}
holds if and only if
\begin{equation}\label{eq-ODEonGraph}
-\n\g_{jn}'' -k^2\g_{jn} = (I\f)_{jn}\quad \mathrm{on }\;\, \L_{jn},
\end{equation}
 and
\begin{equation}\label{eq-approximate-eigen-eq}
(H_2-k^2+M_1)K\g =(1+M_2)\f.
\end{equation}
Moreover, $M_1$ and $M_2$ satisfy $\norm{(H_2-k^2)^{-1}M_1}=\mathcal{O}(\ell)$ and $\norm{(H_2-k^2)^{-1}M_2}=\mathcal{O}(\ell)$ as $\ell\to 0$.
\end{lem}
\begin{proof}
We note \eqref{eq-H1-eigen-eq} implies \eqref{eq-ODEonGraph} by the definition of $H_1$. We denote $\f_{jn}=(I\f)_{jn}$ and recall that
\[
\f_{jn}(x)= \biggpare{1-\frac{x}{\ell}}\f_j+\frac{x}{\ell} \f_n= \f_j+\frac{x}{\ell}(\f_n-\f_j),
\quad x\in[0,\ell]\cong \L_{jn}.
\]
Given the boundary values $\g_j=\g_{jn}(0)$ and $\g_n=\g_{jn}(\ell)$, we can solve the equation \eqref{eq-ODEonGraph} explicitly using the standard ODE method, obtaining thus the expression
\begin{align*}
\g_{jn}(x)
&= \frac{\sin(k'x)}{\sin(k'\ell)}\g_n +\frac{\sin(k'(\ell-x))}{\sin(k'\ell)}\g_j \\[.2em]
&\quad +\frac{1}{k'^2}\biggpare{\frac{\sin(k'x)}{\sin(k'\ell)}-\frac{x}{\ell}}\frac{\f_n}{\n}
+\frac{1}{k'^2}\biggpare{\frac{\sin(k'(\ell-x))}{\sin(k'\ell)}-1+\frac{x}{\ell}}\frac{\f_j}{\n},
\end{align*}
where $k'=k/\sqrt{\n}$. In particular, this yields
\begin{align*}
\g_{jn}'(j)=\g_{jn}'(0)
&= \frac{k'}{\sin(k'\ell)}(\g_n-\g_j) +\frac{k'(1-\cos(k'\ell))}{\sin(k'\ell)}\,\g_j \\[.2em]
&\quad +\frac{1}{k'^2}\biggpare{\frac{k'}{\sin(k'\ell)}-\frac{1}{\ell}}\frac{\f_n-\f_j}{\n}
+\frac{1-\cos(k'\ell)}{k'\sin(k'\ell)}\frac{\f_j}{\n}.
\end{align*}
Substituting this into the boundary condition in the definition of $H_1$,
\[
\sum_{n\in\V(j)} \g_{jn}'(j)=\a_j\g_j,
\]
we get the relation
\begin{align*}
&\frac{k'}{\sin(k'\ell)}\sum_{n\in\V_j} (\g_n-\g_j) +\frac{k'(1-\cos(k'\ell))}{\sin(k'\ell)}|\V_j| \g_j \\[.2em]
&+\frac{1}{k'^2}\biggpare{\frac{k'}{\sin(k'\ell)}-\frac{1}{\ell}}\sum_{n\in\V_j}\frac{\f_n-\f_j}{\n}
+\frac{1-\cos(k'\ell)}{k'\sin(k'\ell)}|\V_j| \frac{\f_j}{\n}=\a_j\g_j
\end{align*}
for each $j\in\V$. Recalling that $\a_j=\ell V_j$ and $|\V_j|=2\n$, we can rewrite it as
\begin{align*}
&-\frac{1}{\ell^2}\sum_{n\in\V_j} (\g_n-\g_j) +\biggpare{\frac{\sin(k'\ell)}{k'\ell}}V_j \g_j
-k^2  \biggpare{\frac{1-\cos(k'\ell)}{(k'\ell)^2/2}}\g_j \\[.2em]
&\quad = -\biggpare{\frac{\sin(k'\ell)-k'\ell}{(k'\ell)^3}}\sum_{n\in\V_j}\frac{\f_n-\f_j}{\n}
+\biggpare{\frac{1-\cos(k'\ell)}{(k'\ell)^2/2}}\f_j.
\end{align*}
Next we note that by the Taylor series expansion we have
\[
\frac{\sin(k'\ell)}{k'\ell}=1+\mathcal{O}(\ell^2), \quad \frac{1-\cos(k'\ell)}{(k'\ell)^2/2}=1+\mathcal{O}(\ell^2),
\quad \frac{\sin(k'\ell)-k'\ell}{(k'\ell)^3}=\mathcal{O}(1)
\]
as $\ell\to 0$, hence setting
\begin{align*}
&(M_1\g)_j:= \biggpare{\frac{\sin(k'\ell)}{k'\ell}-1} V_j \g_j
-k^2 \biggpare{\frac{1-\cos(k'\ell)}{(k'\ell)^2/2}-1}\g_j, \\[.2em]
&(M_2\f)_j:= -\frac{1}{\n} \biggpare{\frac{\sin(k'\ell)-k'\ell}{(k'\ell)^3}}\sum_{n\in\V_j}(\f_n-\f_j)
+\biggpare{\frac{1-\cos(k'\ell)}{(k'\ell)^2/2}-1}\f_j,
\end{align*}
we can rewrite the above relation in the form\eqref{eq-approximate-eigen-eq},
\[
(H_2-k^2+M_1)K\g =(1+M_2)\f.
\]
Finally, we use the following claim to conclude the proof.
\begin{lem}\label{lem-relative-bound}
Suppose $V$ is bounded from below. Then for each $z\in \co\setminus\re$, there is a $C>0$ such that
\[
\norm{\triangle_d(H_2-z)^{-1}}_{\mathcal{B}(\mathcal{H}_2)}\leq C \ell^{-1}, \quad
\norm{V(H_2-z)^{-1}}_{\mathcal{B}(\mathcal{H}_2)}\leq C \ell^{-1}
\quad\mathrm{for }\;\: 0<\ell\leq 1.
\]
\end{lem}

\medskip

Using this result in combination with the above explicit expressions of $M_1$ and $M_2$ we get the estimates $\bignorm{(H_2-z)^{-1}M_1} =\mathcal{O}(\ell)$ and $\bignorm{(H_2-z)^{-1}M_2}=\mathcal{O}(\ell)$ as $\ell\to 0$ for any $z\in\co\setminus\re$.
\end{proof}

\begin{proof}[Proof of Lemma~\ref{lem-relative-bound}]
Since $V$ is by assumption bounded from below, $\triangle_d$ and $V$ are relatively form bounded with respect to $H_2=-\triangle_d+V$. On the other hand, we note that $\norm{\triangle_d}_{\mathcal{B}(\mathcal{H}_2)}=2\n\ell^{-2}$, and therefore 
\[
\norm{\triangle_d(H_2-z)^{-1}}\leq
\norm{|\triangle_d|^{1/2}}\cdot
\norm{|\triangle_d|^{1/2}(H_2-z)^{-1}}
\leq C\ell^{-1}.
\]
Then we also have
\[
\norm{V(H_2-z)^{-1}} = \norm{(H_2+\triangle_d)(H_2-z)^{-1}}
\leq \norm{H_2(H_2-z)^{-1}}+C\ell^{-1},
\]
and this completes the proof.
\end{proof}

\subsection{Approximation theorem}

Now we are in position to compare the resolvents of the operators $\n H_1$ and $H_2$ using the identification map $I$.
\begin{thm}\label{thm-key-thm}
Let $z\in \co\setminus\re$,  then there is a $C>0$ such that
\begin{align*}
&\bignorm{(H_2-z)^{-1} - I^* (\n H_1-z)^{-1} I}_{\mathcal{B}(\H_2)}\leq C\ell, \\[.2em]
&\bignorm{(\n H_1-z)^{-1} - I(H_2-z)^{-1} I^*}_{\mathcal{B}(\H_1)}\leq C\ell.
\end{align*}
\end{thm}
\begin{proof}
We recall that $z=k^2\in\co\setminus\re$. By Lemma~\ref{lem-key-ODE-estimate}, we have
\begin{equation}\label{eq-2-3-1}
(H_2-z+M_1)K(\n H_1-z)^{-1}I =1+M_2
\end{equation}
on $\H_2$. We use the identity
\[
H_2-z+M_1= (H_2-z)\bigpare{1+(H_2-z)^{-1}M_1}
\]
which implies
\[
(H_2-z+M_1)^{-1}= \bigpare{1+(H_2-z)^{-1}M_1}^{-1} (H_2-z)^{-1}
\]
as long as $\ell$ is sufficiently small so that the first factor on the right-hand side makes sense. Combining this with \eqref{eq-2-3-1}, we get
\[
K(\n H_1-z)^{-1}I= \bigpare{1+(H_2-z)^{-1}M_1}^{-1} (H_2-z)^{-1}(1+M_2),
\]
and therefore
\begin{align*}
&K(\n H_1-z)^{-1}I-(H_2-z)^{-1} \\[.2em]
&=-\bigpare{1+(H_2-z)^{-1}M_1}^{-1} \bigpare{(H_2-z)^{-1}M_1 (H_2-z)^{-1} -(H_2-z)^{-1}M_2}.
\end{align*}
This implies, again by virtue of Lemma~\ref{lem-key-ODE-estimate},
\[
\bignorm{(H_2-z)^{-1}- K(\n H_1-z)^{-1}I}_{\mathcal{B}(\H_2)}=\mathcal{O}(\ell) \quad \text{as }\;\ell\to 0.
\]
Now we use Lemma~\ref{lem-I^*-K-estimate} and the triangle inequality to conclude that
\[
\bignorm{(H_2-z)^{-1}- I^*(\n H_1-z)^{-1}I}_{\mathcal{B}(\H_2)}=\mathcal{O}(\ell) \quad \text{as }\;\ell\to 0,
\]
since $(\n H_1-z)^{-1}$ is bounded as a map from $\H_1$ to $H^1(\Gamma)$.
In a similar way, we use Lemma~\ref{lem-IK-1-estimate} and Corollary~\ref{cor-II^*-1-estimate} to get the other estimate,
\[
\bignorm{I(H_2-z)^{-1}I^*- (\n H_1-z)^{-1}}_{\mathcal{B}(\H_1)}=\mathcal{O}(\ell)
\]
as $\ell\to 0$. This completes the proof.
\end{proof}

\section{Proof of Theorem~\ref{thm-main}}

To finish the task, it is now sufficient to combine Theorem~\ref{thm-key-thm} with Theorem~\ref{thm-NT}. We define the identification operator by
\[
\G:=\F I^*.
\]
Using the fact that $\F$ is bounded, in fact an isometry, we then have
\begin{align*}
\bignorm{&(H-z)^{-1} - \G(\n H_1-z)^{-1} \G^*} \\[.2em]
&\leq \bignorm{(H-z)^{-1} - \F(H_2-z)^{-1} \F^*}
+\bignorm{\F(H_2-z)^{-1}\F^* - \G(\n H_1-z)^{-1} \G^*} \\[.2em]
&=\bignorm{(H-z)^{-1} - \F(H_2-z)^{-1} \F^*}
+\bignorm{\F\bigpare{(H_2-z)^{-1} - I^*(\n H_1-z)^{-1}I }\F^*} \\[.2em]
&\leq \bignorm{(H-z)^{-1} - \F(H_2-z)^{-1} \F^*}
+\bignorm{(H_2-z)^{-1} - I^*(\n H_1-z)^{-1}I} \\[.2em]
&\to 0 \quad \text{as }\;\ell\to 0.
\end{align*}
The proof of the other estimate is almost identical, so we omit the computation; by that the proof of Theorem~\ref{thm-main} is finished. \qed

\subsection*{Acknowledgements}

P.E. was supported by the Czech Science Foundation within the project 21-07129S and by the EU project ${\rm CZ}.02.1.01/0.0/0.0/16$\underline{ }$ 019/0000778$.
S.N. was partially supported by JSPS Grant Numbers 15H03622 (2015--2019) and 21K03276 (2021--2024).
Y.T. was partially supported by JSPS Grant Numbers 20J00247 (2020--2021) and 21K20337 (2021--2023)


\end{document}